\documentclass[twoside,a4paper,11pt]{amsart}

\usepackage{amsmath,amssymb,dsfont}
\usepackage[english]{babel}
\usepackage{mathrsfs}
\usepackage{hyperref}

\newcommand{\beq}{\begin{equation}}
\newcommand{\eeq}{\end{equation}}
\newcommand{\tr}{\operatorname{Tr}\,}
\renewcommand{\Im}{\operatorname{Im}\,}
\renewcommand{\Re}{\operatorname{Re}\,}

\newcommand{\sgn}{\operatorname{sgn}}

\newcommand{\RE}{\mathbb R}
\newcommand{\CO}{\mathbb C}
\newcommand{\NO}{\mathbb N}

\newcommand{\ZA}{\mathbb Z}

\newcommand{\WW}{W}

\newcommand{\ga}{\gamma}

\newcommand{\al}{\alpha}
\newcommand{\la}{\lambda}

\renewcommand{\H}{{\mathcal{H}}}
\newcommand{\ka}{{\kappa}}

\newcommand{\Sp}{\rm Sp}

\newcommand{\e}{{\rm e}}
\newcommand{\p}{{\partial}}

\newcommand{\Z}{{\mathds{Z}}}
\newcommand{\R}{{\mathds{R}}}

\DeclareMathOperator*{\Rz}{Res_0}
\DeclareMathOperator*{\Ru}{Res_1}
\DeclareMathOperator*{\Rd}{Res_2}
\DeclareMathOperator*{\Rk}{Res_k}

\newtheorem{theorem}{Theorem}[section]
\newtheorem{corollary}[theorem]{Corollary}
\newtheorem{lemma}[theorem]{Lemma}
\newtheorem{proposition}[theorem]{Proposition}

\theoremstyle{definition}

\newtheorem{remark}[theorem]{Remark}

\numberwithin{equation}{section}

\renewcommand{\gg}{g}

\title[Relative partition function of Coulomb plus delta interaction]{Relative partition function of Coulomb plus delta interaction}

\author[S. Albeverio]{Sergio Albeverio}
\address{Institut f\"{u}r Angewandte Mathematik, Universit\"{a}t Bonn,  Endenicher Allee 60,  53115 Bonn, Germany;
SFB611, Bonn; HCM, Bonn; BIBOS, Bielefeld
and Bonn; IZKS, Bonn; CERFIM, Locarno;
Acc. Arch., USI, Mendrisio; Dip. Matematica Universit\`a di Trento, Italy}
\email{albeverio@uni-bonn.de}

\author[C. Cacciapuoti]{Claudio Cacciapuoti}
\address{Dipartimento di Scienza e Alta Tecnologia, Universit\`a dell'Insubria, Via Valleggio 11, 22100 Como, Italy}
\email{claudio.cacciapuoti@uninsubria.it}

\author[M. Spreafico]{Mauro Spreafico}
\address{Dipartimento di Matematica e Fisica Ennio De Giorgi,  Universit\`a del Salento \& INFN, Via Arnesano, 73100 Lecce,
Italy}
\email{mauro.spreafico@unisalento.it}

\begin{document}

\begin{abstract}
The relative partition function and the relative zeta function of the perturbation of the Laplace operator by a Coulomb potential  plus a point interaction centered  in the origin  is discussed. Applications to the study of the Casimir effect are indicated. 
\end{abstract}



\maketitle

{\it Dedicated to Pavel Exner}

\section{Introduction}
\label{s0}

The present paper discusses a problem related to three main areas of investigations, in mathematics and physics: the theory of quantum fields (in particular thermal fields), the study of determinants of elliptic (pseudo differential) operators, and the study of singular perturbations of  linear operators. The problem providing the link between these areas originated with a theoretical investigation by H. B. G.  Casimir \cite{casimir} who predicted the possibility of an effect, called ``Casimir effect'', of attraction of parallel conducting plates in vacuum due to the presence of fluctuations in the vacuum energy of the electromagnetic quantum field.

Since the experimental confirmation of this effect by Spaarnay \cite{spaarnay}, about ten years after the work of Casimir, both theoretical and experimental studies of ``Casimir like effects'' have received a lot of attention. In particular the temperature corrections  where first discussed by  M. Fierz \cite{Fie} and J. Mehra \cite{Meh}, we refer to the monograph  \cite{BKMM} for more references and details on the effects of temperature.  On the other hand,  its dependence on the geometry of the plates and the medium (even attractiveness can become repulsion according to changing geometry) has been discussed in several publications, see, e.g., the books \cite{BKMM,BEORZ,Elizalde1,Elizalde,Milton,MT}, the survey papers \cite{Bordagetal,PMG}, and, e.g., \cite{Boi,BPN,BV,Bressietal,Cas,Cognolaetal,CFNT,DoKe,Do,Duk,EVZ,Lamoureaux,MuC,orSpr,Pow,Scandurra}.

The physical discussion of the Casimir effect is also related to the one of the Van der Waals forces between molecules, see \cite{MT}. It has also many relations to condensed matter physics, hadronic physics, cosmology, and nanotechnology, see, e.g.,  the references in \cite{Bordagetal,BKMM,Elizalde1,Elizalde,Milton,MT,PMG}.

Theoretically the Casimir effect arises when computing the difference between two infinite quantities, namely the vacuum  energy of a quantum field with or without a certain ``boundary condition''. More generally it is a phenomenon related to the difference of two Green's functions associated with hyperbolic or elliptic operators. Such problems are also of interest in geometric analysis, particularly since the work by W. M\"uller \cite{Muller}, and M. Spreafico and S. Zerbini \cite{SZ}. The latter works are related to the introduction by Ray and Singer \cite{RS} of a definition of determinants for elliptic operators on manifolds via a zeta-function renormalization (see also, e.g., \cite{Lesch,NLS04}). By this procedure one can define $\log(\det A )^{-\frac12} $, for $A$ self adjoint, positive, in some Hilbert space, via the analytic continuation at $s= \frac12$ of the zeta-function associated with $A$, defined for $\Re s$ sufficiently large as
\[
\zeta(s;A) := \sum_{\la\in\sigma^+(A)} \la^{-s}
\]
$\sigma^+(A)$ being the positive part of the spectrum of $A$. Setting
\[Z:=(\det A)^{-\frac12},\]
 one has the  definition of the ``partition function''
 \[
 \text{``} \; Z =\int_{\Phi} e^{- S(\varphi)} d\varphi \;  \text{''},
 \]
$S(\varphi):= (\varphi, A\varphi)$, associated with a (Euclidean) quantum field with covariance operator given by the inverse of $A$ ($\varphi$ is the field, $\Phi$ the space of ``fields configurations'').

In turn, it is well known that partitions functions $Z$ arise as normalizations in heuristic Euclidean path integrals
 \[
 \text{``} \; Z^{-1} \int_{\Phi} e^{- S(\varphi)} f( \varphi) d\varphi \; \text{''},
 \]
$f$ being complex valued functions (related to ``observables''), see, e.g., \cite{0,Sym}.

On the other hand it was pointed out by Hawking \cite{Hawking} and, independently, Figari, H{\o}egh-Krohn, and Nappi \cite{FHKN}, that   there is  a strict relation between Euclidean vacuum states in de Sitter spaces of  fixed curvature and temperature states of Euclidean states. Hawking used the Ray-Singer definition of a partition function related to $A$ to compute physical quantities of the Euclidean model. For wide-ranging extensions of these connections see, e.g., \cite{AsIM,AtPaSi,FeP1,FeP2,Ful,Gibbons,Gilkey,Mo,Muller,Spreafico1,S06}.

Another application of the zeta function is in the computation of the high temperature asymptotics of several thermodynamic functions such as  the Helmholtz free energy, internal energy, and entropy, see, e.g.,  \cite{BNP02} and references therein.

As pointed out in \cite{Muller} and \cite{Spreafico2,S06,66a}, considering the relative zeta-function of a pair of elliptic operators $A$, $A_0$, leads to define, via a relative zeta-function, a relative determinant including $A$ and $A_0$ and a Casimir effect can be discussed relatively to the pair $(A, A_0)$. In fact, the strength of the Casimir effect is expressed by the derivative of the relative zeta-function at $0$. These considerations are also related to the study of relative traces of  semigroups resp. resolvents associated with pairs of operators. The study of such relative traces has its origins in quantum statistical mechanics \cite{Beth-Uhl}.

The case where $A_0$ is the Laplace-Beltrami operator on $S^1\times\RE^3$, and $A$ is a point perturbation of $A_0$ has been discussed in details in \cite{SZ} and \cite{ACSZ}. For the extended study of point interactions on $\RE^d$, $d=1,2,3$, see \cite{0, AGH-KH05,AK}. The case where $\RE^d$ is replaced by a Riemannian manifold  occurs particularly in \cite{ColindeVerdiere} (who points out its possible relevance in number theory), see also \cite{CQ,FeP2,Kurokawa}.

For further particular studies of point interactions in relation with the Casimir effect see \cite{ABD,AGH-KH05,BPN,BV,GKV,GJKQSW,Milton,Khusnutdinov,Park,ST,Solodukhin}.

Particularly close to our work is the result in \cite{ACSZ} where $A_0$ is the half space $x^3>0$ in $\RE^3$ and $A$ is taken to be the sum of two point interactions located at $(a_1,a_2,a_3)$ and $(a_1,a_2,-a_3)$, $a_1,a_2\in\RE$, $a_3\in\RE_+$. The relative trace of the resolvents was computed at values of the spectral parameter $\la$ such that  $\Im \sqrt\la >0$, and the spectral measure was constructed. Moreover the asymptotics for small and large values of the spectral parameter was found. Furthermore the relative zeta-function and its derivative at $0$ has been computed and related to the Casimir effect \cite{ACSZ}.

The present paper extends this kind of relations to the case of the pair $(A,A_0)$, where $A_0$ is the operator $-\Delta$ with a Coulomb interaction at the origin acting in $L^2(\RE^3)$, and $A$ is a perturbation of $A_0$ obtained by adding a point interaction at the origin. The construction of $A_0$ and $A$ is based on \cite{AGH-KH05}, Ch I.2. In order to define and study the relative partition function we use explicit formulae for the integrals of the Whittaker's functions which enter the explicit expression of the resolvent of $-\Delta$ with a Coulomb interaction.

Such explicit  formulae do not exist  in the situation where the point interaction is not centered at the origin.  In this situation an alternative approach would be to use series expansions to compute the integrals. It turns out that  this idea does not seem feasible due to the slow decay of the Coulomb interaction at infinity. On the other hand, the case of potentials with faster decay at infinity should be treatable in this way, replacing the explicit formulae  by methods of regular perturbations theory.

The structure of the paper is as follows. In Section \ref{s2} we recall the general definition of the relative partition function associated to a pair of non-negative self-adjoint operators and its relation with the relative zeta function. In Section \ref{s1} we study the perturbation of the Laplacian by a Coulomb and a delta potential centered at the origin. In Section \ref{s3} we study the associated relative partition function of the Coulomb plus delta interaction.

%
%

\section{Relative partition function associated to a pair of non-negative self adjoint operators}
\label{s2}


This section presents a generalization of  the method introduced in \cite{SZ} to study the analytic properties of the relative zeta function associated to a pair of operators $(A,A_0)$ as described below (see also \cite{Muller}). We assume here that logarithmic terms appear in the expansion of the relative trace, and this  will produce a double pole in  the relative zeta function, and in turn  a simple pole in the  relative partition function.


\subsection{Relative  zeta function}

We denote by $R(\la;A) \equiv (\la-A)^{-1}$ the resolvent of a linear operator $A$. $\la$ is in the resolvent set, $\rho(A)$, of $A$, a subset of $\CO$. The relative zeta function $\zeta(s;A,A_0)$ for a pair of non negative self adjoint operators $(A,A_0)$ is defined when the relative resolvent $R(\lambda;A)-R(\lambda; A_0)$ is of trace class and  some conditions on the asymptotic expansions of the trace of the relative resolvent $r(\lambda; A,A_0)$ are satisfied, as in Section 2 of \cite{SZ}. These conditions imply that similar conditions on the trace of the relative heat operator $ \tr \left(\e^{-t A}-\e^{-t A_0}\right)$ are satisfied, according to Section 2 of \cite{Muller}. The conditions in \cite{SZ} on the asymptotic expansions ensure that the relative zeta function is regular at $s=0$. In the present work  we consider a wider class of pairs, and we admit a more general type of asymptotic expansions, as follows. Let $\H$ be a separable Hilbert space, and let $A$ and $A_0$ be two self adjoint non negative linear operators in $\H$. Suppose that $\Sp A=\Sp_c A$, is purely continuous, and  assume both $0$ and $\infty$ are accumulation points of $\Sp A$. 

Then, by a standard argument (see for example  the proof of the corresponding result in \cite{SZ}), we prove  Lemma \ref{l2.1} below.

Let us  recall first the definition of asymptotic  expansion. If $f(\la)$ is a complex valued function, we write  $f(\la)\sim \sum_{n=0}^\infty a_n\la^n$, $a_n\in\CO$, $\la\to0$, if for any $N\in\NO_0$ one has  $\frac{f(\la)- \sum_{n=0}^N a_n\la^n}{\la^N}\to0$ as $\la\to 0$, and we say that $f$ has the asymptotic expansion  $\sum_{n=0}^\infty a_n\la^n$. Then the following result holds true:

\begin{lemma}
\label{l2.1}
Let $(A,A_0)$ be a pair of non negative self adjoint operators as above  satisfying the following conditions:
\begin{enumerate}
\item[(B.1)] the operator $R(\lambda;A)-R(\lambda;A_0)$ is of trace class for all $\lambda\in \rho(A)\cap\rho(A_0)$;
\item[(B.2)] as $\lambda\to \infty$ in $\rho(A)\cap\rho(A_0)$, there exists an asymptotic expansion of the form:
\[
\tr (R(\lambda;A)-R(\lambda;A_0))\sim \sum_{j=0}^\infty \sum_{k=0}^{K_j} a_{j,k} (-\lambda)^{\alpha_j} \log ^k (-\lambda),
\]
where $a_{j,k}\in \CO$,  $-\infty<\dots<\alpha_1<\alpha_0$, $\alpha_j\to -\infty$, for large $j$; 
\item[(B.3)] as $\lambda\to 0$, there exists an asymptotic expansion of the form
\[
\tr (R(\lambda;A)-R(\lambda;A_0))\sim \sum_{j=0}^\infty b_{j} (-\lambda)^{\beta_j};
\]
where $b_j\in\CO$, $-1\leq\beta_0<\beta_1<\dots$, and $\beta_j\to +\infty$, for large $j$,
\end{enumerate}

\begin{enumerate}
\item[(C)] $\alpha_0<\beta_0$.
\end{enumerate}

Then the relative zeta function is defined by
\[
\zeta(s;A,A_0)=\frac{1}{\Gamma(s)}\int_0^\infty t^{s-1} \tr\left(\e^{-tA}-\e^{-tA_0 }\right)dt,
\]
when $\alpha_0+1<\Re(s)<\beta_0+1$, and by analytic continuation elsewhere. Here $\Gamma$ is the classical Gamma function and
\[
\tr\left(\e^{-tA}-\e^{-tA_0 }\right)=\frac{1}{2\pi i} \int_{\Lambda}\e^{-\lambda t} \tr (R(\lambda; A)-R(\lambda;A_0)) d\lambda,
\]
where $\Lambda$ is some contour of Hankel type (see, e.g., \cite{Erd54, S06}).    The analytic extension of $\zeta(s;A,A_0)$ is regular except for possible simple poles at $s=\beta_j$ and possible further  poles at $s=\alpha_j$.
\end{lemma}

Note that the poles of the relative zeta function at $s=\alpha_j$ can be of higher orders, differently from the case investigated in \cite{SZ}.


Introducing the relative spectral measure, we have the following useful representation  of the relative zeta function.

\begin{lemma}\label{z1} Let $(A,A_0)$ be a pair of non negative self adjoint operators as above satisfying conditions (B.1)-(B.3), and (C) of Lemma \ref{l2.1}. Then,
\[
\zeta(s;A,A_0) =\int_0^\infty v^{-2s}e(v;A,A_0) dv ,
\]
where the relative spectral measure is defined by
\begin{align}
\label{spm}e(v;A,A_0)&=\frac{v}{\pi i}\lim_{\epsilon\to 0^+}\left(r(v^2\e^{2i\pi-i\epsilon};A,A_0)-r(v^2\e^{i\epsilon};A,A_0)\right)\qquad v\geq 0,\\
\label{rrr}r(\lambda;A,A_0)&=\tr(R(\lambda;A)-R(\lambda;A_0)) \qquad \la\in\rho(A)\cap\rho(A_0).
\end{align}

The integral, the limit and the trace exist.
\end{lemma}
\begin{proof} Since $(A,A_0)$ satisfies (B.1)-(B.3), we can write
\[
\tr\left(\e^{-tA }-\e^{-tA_0 }\right)=\frac{1}{2\pi i}\int_{\Lambda} \e^{-\lambda t} \tr(R(\lambda;A)-R(\lambda;A_0))d\lambda.
\]

Changing the spectral variable $\la$ to $k=\lambda^{\frac{1}{2}}$, with the principal value of the square root, i.e. with $0<\arg k<\pi$, we get
\[
{\rm Tr}\left(\e^{-tA}-\e^{-tA_0  }\right)=\frac{1}{\pi i}\int_\gamma \e^{-k^2 t} \tr(R(k^2;A)-R(k^2;A_0))k d k,
\]
where $\gamma$ is the line $k=-ic$, for some $c>0$. Writing
$k=v\e^{i\theta}$, $0\leq\theta<2\pi$, and $r(\lambda;A,A_0)=\tr(R(\lambda;A)-R(\lambda;A_0))$, a standard computation leads to
\begin{align*}
{\rm Tr}\left(\e^{-tA }-\e^{-tA_0 }\right)&=\int_0^\infty \e^{-v^2 t}e(v;A,A_0) dv,\\
\zeta(s;A,A_0) &=\int_0^\infty v^{-2s}e(v;A,A_0) dv .
\end{align*}
\end{proof}

\begin{remark}
The relative spectral measure is discussed in general, e.g., in \cite{Muller}. It is expressed by \eqref{rrr} in terms of $r(\lambda;A,A_0)$ which is the Laplace transform of $\tr\left(\e^{-tA }-\e^{-tA_0 }\right)$, which in turn is simply related to the spectral shift function (see Eq. (0.6) in \cite{Muller}). The derivative of the latter is essentially the density of states used, e.g., in \cite{NP12} in connection with the Casimir effect, and going back to the original work by M.G. Krein and M.Sh. Birman \cite{birman-krein, krein1,krein2}.
 \end{remark}

It is clear by construction that the analytic properties of the relative zeta function are determined by the asymptotic expansions required in conditions (B.1) and (B.2). More precisely, such conditions imply similar conditions on the expansion of the relative spectral measure, and hence on the analytic structure of the relative zeta function. This is in the next lemmas.

\begin{lemma}\label{expmean} Let $(A,A_0)$ be a pair of non negative self adjoint operators as in Lemma \ref{z1}. Then the relative spectral measure $e(v;A,A_0)$ has the following asymptotic expansions. For small  $v\geq0$:
\[
e(v;A,A_0)\sim\sum_{j=0}^\infty c_j v^{2\beta_j+1},
\]
where
\[
c_j=-\frac{2b_j \sin \pi\beta_j}{\pi},
\]
and the $\beta_j$ and the $b_j$ are the numbers appearing in condition (B3) of Lemma \ref{l2.1}; for large $v\geq 0$ and $j\in\NO_0$:
\begin{align*}
e(v;A,A_0)&\sim \sum_{j=0}^\infty \sum_{h=0}^{H_j}  e_{j,h}  v^{2\alpha_j+1}\log ^h v\\
&\sim\sum_{j=0}^\infty \sum_{k=0}^{K_j} \sum_{h=0}^k e_{j,k,h}  v^{2\alpha_j+1}\log ^h v^2 ,
\end{align*}
where
\[
e_{j,k,h}=-  a_{j,k} (\pi i)^{k-h-1}\binom{k}{h}\left(\e^{i\alpha_j\pi}-(-1)^{k-h}\e^{-i\alpha_j\pi}\right),
\]
and the $a_{j,k}$, $\alpha_j$,  and $K_j$ are the numbers appearing in condition (B2) of Lemma \ref{l2.1}. The coefficients $e_{j,h}$ can be expressed in terms of the coefficients $e_{j,k,h}$.

\end{lemma}
\begin{proof} Note that the cut $(0,\infty)$ in the complex $\lambda$-plane corresponds to the cut $(-\infty,0)$ in the complex $-\lambda$-plane. Thus $-\lambda=x\e^{i\theta}$, with $-\pi\theta<\pi$, and $\theta=0$ corresponds to positive real  values of $-\lambda$.

Thus, inserting the expansion (B3) for small $\lambda$ in the definition of the relative spectral measure, equation (\ref{spm}), we obtain,  for small $v$,
\[
e(v;A,A_0)\sim-\frac{v}{i\pi}\lim_{\epsilon\to 0^+} \sum_{j=0}^\infty b_j  v^{2\beta_j}\left( \e^{(\pi i-i\epsilon) \beta_j}-\e^{(-\pi i+i\epsilon) \beta_j}\right),
\]
and the first part of the statement follows. For the expansion for large $v$, we insert (B2) into the definition of the relative spectral measure. This gives, for large $v$,
\begin{align*}
&e(v;A,A_0) \\
\sim&-\frac{v}{i\pi}\lim_{\epsilon\to 0^+} \sum_{j=0}^\infty \sum_{k=0}^{K_j} a_{j,k}  v^{2\alpha_j} \\
& \qquad\qquad \times \left( \e^{(\pi i-i\epsilon) \alpha_j}(\log v^2+\pi i-i\epsilon)^k -\e^{(-\pi i+i\epsilon) \alpha_j}(\log v^2-\pi i+i\epsilon)^k\right),\\
&e(v;A,A_0) \\
\sim & - \sum_{j=0}^\infty v^{2\alpha_j+1} \sum_{k=0}^{K_j} \frac{a_{j,k}}{i\pi}  \\
& \qquad\qquad \times
\left( \e^{\pi i \alpha_j}\sum_{h=0}^k \binom{k}{h} (i\pi)^{k-h}\log^k v^2 -\e^{-\pi i \alpha_j}\sum_{h=0}^k \binom{k}{h} (-i\pi)^{k-h}\log^k v^2 \right),
\end{align*}
and the thesis follows. \end{proof}

\begin{remark} We give more details on the first coefficients  that are more relevant in the present work. Direct calculation gives
\begin{align*}
e_{j,0}&=e_{j,0,0}+2\sum_{k=1}^{K_j} e_{j,k,0}=-\sum_{k=0}^{K_j}a_{j,k}(\pi i)^{k-1}\left(\e^{i\alpha_j \pi}-(-1)^k \e^{-i\alpha_j \pi}\right),\\
e_{j,0,0}&=-\frac{2\sin\pi\alpha_j}{\pi}a_{j,0}.
\end{align*}
\end{remark}

\subsection{Relative partition function}

Let $W$ be a smooth Riemannian manifold of dimension $n$, and
consider the product $X=S^1_\frac{\beta}{2\pi}\times W$, where
$S^1_r$ is the circle of radius $r$, $\beta>0$. Let $\xi$ be a complex line
bundle over $X$, and $L$ a self adjoint non negative linear operator
on the Hilbert space $\mathcal{H}(W)$ of the $L^2$ sections of the
restriction of $\xi$ onto $W$, with respect to some fixed metric $g$
on $W$. Let $L$ be the self adjoint non negative operator
$L=-\p^2_u+A$, on the Hilbert space $\mathcal{H}(X)$ of the $L^2$ sections of
$\xi$, with respect to the product metric  $du^2\oplus g$  on $X$,
and with periodic boundary conditions on the circle. Assume that
there exists a second operator $A_0$ defined on $\mathcal{H}(W)$, such that
the pair $(A,A_0)$ satisfies the assumptions (B.1)-(B.3) of Lemma
\ref{l2.1}. Then, by a proof similar to the one of Lemma 2.1 of \cite{SZ}, it is possible to show that there exists a second operator $L_0$
defined in $\mathcal{H}(X)$, such that the pair $(L,L_0)$ satisfies those
assumptions too. Under these requirements, we define the regularized relative
zeta  partition function of the model described by the
pair of operators $(L,L_0)$ by
\beq\label{partpart}
\log Z_{\rm R}=\frac{1}{2}\Rz_{s=0}\zeta'(s;L,L_0)-\frac{1}{2}\Rz_{s=0}\zeta(s;L,L_0)\log \ell^2,
\eeq
where $\ell$ is some renormalization constant (introduced by Hawking \cite{Hawking}, see also, e.g., \cite{Mo}, in connection with the scaling behavior in path integrals in curved spaces), and we have the following result, in which $\log Z_R$ is essentially expressed  in terms of the relative Dedekind  eta function $\eta(\beta;A,A_0)$.
\begin{proposition}
\label{p1}
Let $A$ be a non negative self adjoint operator on $W$, and $L=-\p^2_u+A$, on $S^1_\frac{\beta}{2\pi}\times W$ as defined above. Assume there exists an operator $A_0$ such that the pair $(A,A_0)$ satisfies conditions (B.1)-(B.3) of Lemma \ref{l2.1}. Then, the relative zeta function $\zeta(s;L,L_0)$ (defined analogously to the one given in Lemma \ref{l2.1}) has a simple pole at $s=0$ with residua:
\begin{align*}
\Ru_{s=0}\zeta(s;L,L_0)=&-\beta \Rd_{s=-\frac{1}{2}}\zeta(s;A,A_0),\\
\Rz_{s=0}\zeta(s;L,L_0)=&-\beta\Ru_{s=-\frac{1}{2}}
\zeta(s;A,A_0)-2\beta(1-\log 2)\Rd_{s=-\frac{1}{2}}
\zeta(s;A,A_0),\\
\Rz_{s=0}\zeta'(s;L,L_0)=&-\beta \Rz_{s=-\frac{1}{2}}
\zeta(s;A,A_0)-2\beta(1-\log 2)\Ru_{s=-\frac{1}{2}}
\zeta(s;A,A_0)\\
&-\beta \left(2+\frac{\pi^2}{6}+2(1-\log 2)^2\right)\Rd_{s=-\frac{1}{2}}
\zeta(s;A,A_0)\\
&-2\log\eta(\beta ;A,A_0),
\end{align*}
where $L_0=-\p^2_u+A_0$, and the relative Dedekind eta function is
defined by
\begin{align*}
\log\eta(\tau;A,A_0)&=\int_0^\infty \log\left(1-\e^{-\tau v}\right) e(v;A,A_0) dv, \qquad \tau>0.
\end{align*}
$\displaystyle\Rk_{s=s_0}\zeta(s)$ is understood as the coefficient of the term $(s-s_0)^{-k}$ in the Laurent expansion of $\zeta(s)$ around $s=s_0$.\\
The residua and the integral are finite.
\end{proposition}
\begin{proof} Since $(A,A_0)$ satisfies (B.1)-(B.3),  the  $(L,L_0)$ relative  zeta function $\zeta(s;L,L_0)$ is defined by
\[
\zeta(s;L,L_0)=\frac{1}{\Gamma(s)}\int_0^\infty t^{s-1} {\rm Tr}\big(\e^{-tL}-\e^{-tL_0 }\big)dt,
\]
when $\alpha_0+1<\Re(s)<\beta_0+1$ (with $\alpha_0$ and $\beta_0$ as in Lemma \ref{l2.1}). Since (see for example  Lemma 2.2 of \cite{SZ})
\begin{align*}
{\rm Tr}\left(\e^{-Lt}-\e^{-L_0t}\right)=&\sum_{n\in \Z}\e^{-\frac{n^2}{r^2}t} {\rm Tr}\left(\e^{-tA }-\e^{-tA_0 }\right),
\end{align*}
where $r=\beta/(2\pi)$ and $t>0$. Using the Jacobi summation formula and dominated convergence to exchange summation and integration
we obtain
\[
\begin{aligned}
\zeta(s;L,L_0)=&\frac{1}{\Gamma(s)}\int_0^\infty t^{s-1} \sum_{n\in \Z}\e^{-\frac{n^2}{r^2} t}{\rm Tr}\left(\e^{-tA}-\e^{-tA_0 }\right)dt\\
=&\frac{\sqrt{\pi}r}{\Gamma(s)}\int_0^\infty t^{s-\frac{1}{2}-1} {\rm Tr}\left(\e^{-tA}-\e^{-tA_0 }\right)dt\\
&+\frac{2\sqrt{\pi}r}{\Gamma(s)}\int_0^\infty t^{s-\frac{1}{2}-1} \sum_{n=1}^\infty\e^{-\frac{\pi^2 r^2 n^2}{t}}{\rm Tr}\left(\e^{-tA}-\e^{-tA_0 }\right)dt\\
=&z_1(s)+z_2(s),
\end{aligned}
\]
with
\[
\begin{aligned}
z_1(s) :=&\frac{\sqrt{\pi}r}{\Gamma(s)}\Gamma\left(s-\frac{1}{2}\right) \zeta\left(s-\frac{1}{2};A,A_0\right),\\
z_2(s) :=&\frac{2\sqrt{\pi}r}{\Gamma(s)}\sum_{n=1}^\infty\int_0^\infty t^{s-\frac{1}{2}-1} \e^{-\frac{\pi^2 r^2n^2}{t}}{\rm Tr}\left(\e^{-tA}-\e^{-tA_0 }\right)dt.
\end{aligned}
\]

The first term, $z_1(s)$, can be expanded near $s=0$, and this gives the result stated, by Lemma \ref{l2.1}. By Lemma \ref{z1}, the second term $z_2(s)$ is
\begin{align*}
z_2(s)=&\frac{2\sqrt{\pi}r}{\Gamma(s)}\sum_{n=1}^\infty\int_0^\infty t^{s-\frac{1}{2}-1} \e^{-\frac{\pi^2n^2r^2}{t}}\int_0^\infty \e^{-v^2 t} e(v;A,A_0) d vdt,\\
\end{align*}
and we can do the $t$ integral using for example \cite[3.471.9]{GraRyz07}.  We obtain
\beq\label{po}\begin{aligned}
z_2(s)=&\frac{4\sqrt{\pi} r}{\Gamma(s)}\sum_{n=1}^\infty\int_0^\infty \left(\frac{\pi n r}{v}\right)^{s-\frac{1}{2}}K_{s-\frac{1}{2}}(2\pi n rv) e(v;A,A_0) d v.\\
\end{aligned}
\eeq

Since the Bessel function $K_{s-\frac{1}{2}}(2\pi n rv)$ is analytic in its parameter, regular at $-\frac{1}{2}$, and $K_{-\frac{1}{2}}(z)=\sqrt{\frac{\pi}{2z}}\e^{-z}$, equation (\ref{po}) gives the formula for the analytic extension of the relative zeta function $\zeta(s;H,H_0)$ near $s=0$. We obtain
\[
z_2(0)=0\,, \quad
z_2'(0)=-2\int_0^\infty \log\left(1-\e^{-2\pi rv}\right)e(v;L,L_0) d v,
\]
and the integral converges by assumptions (B.2) and (B.3).

\end{proof}

It is clear by the previous result that all information on the relative partition function comes from the analytic structure of the spatial relative spectral function $\zeta(s;A,A_0)$ near $s=-\frac{1}{2}$. Such information is based on the asymptotic expansion assumed for the relative resolvent, and contained in the following lemma.

\begin{lemma}\label{rzf} Let $(A,A_0)$ be a pair of non negative self adjoint operators as in Lemma \ref{z1}. 
Then, the relative zeta function $\zeta(s;A,A_0)$ extends analytically to the following meromorphic function in a neighborhood of $s=-\frac{1}{2}$:
\begin{equation} \label{2.5a}
\begin{aligned}
\zeta(s;A,A_0)=&\frac{1}{2}\sum_{j=0}^{J_0-1}\frac{c_j}{\beta_j+1-s}
+\sum_{j=0}^{J_\infty-1}\sum_{h=0}^{H_j} \frac{(-1)^{h+1}e_{j,h}}{2^{h+1} (\alpha_j+1-s)^{h+1}}\\
&+\int_0^1 v^{-2s} \left(e(s;A,A_0)-\sum_{j=0}^{J_0} c_j v^{2\beta_j+1}\right) dv\\
&+\int_1^\infty v^{-2s} \left(e(s;A,A_0)-\sum_{j=0}^{J_\infty} \sum_{h=0}^{H_j} e_{j,h} v^{2\alpha_j+1}\log^h v\right) dv,
\end{aligned}
\end{equation}
where $J_0$ is the smallest integer such that $\beta_{J_0}>-\frac{3}{2}$, and $J_\infty$ is the largest integer such that $\alpha_{J_\infty}<-\frac{3}{2}$ (the $\al_j$, $\beta_j$ resp. $c_j$, $e_{j,h}$, $H_j$ are as in Lemma \ref{l2.1} resp. Lemma \ref{z1}).
\end{lemma}

\begin{proof}  Set
\[
\zeta_0(s;A,A_0)= \int_0^1 v^{-2s} e(v;A,A_0) dv
\]
and
\[
\zeta_\infty(s;A,A_0)=\int_1^\infty v^{-2s} e(v;A,A_0) dv,
\]
then
\[
\zeta(s;A,A_0)=\zeta_0(s;A,A_0)+\zeta_\infty(s;A,A_0).
\]

Consider the expansion of $e(s;A,A_0)$ for small $v$ given in Lemma \ref{expmean}. Let $J_0$ be the smallest integer such that $\beta_{J_0}>-\frac{3}{2}$, and write
\begin{align}
& \zeta_0(s;A,A_0) \nonumber \\
=& \int_0^1 v^{-2s}\left( \sum_{j=0}^{J_0}c_j v^{2\beta_j+1} \right)dv
+\int_0^1 v^{-2s} \left(e(s;A,A_0)-\sum_{j=0}^{J_0} c_j v^{2\beta_j+1}\right) dv.
\label{oo}
\end{align}

The last integral in equation (\ref{oo}) is convergent, while the first one can be computed explicitly. This gives the statement for $\zeta_0$, in the sense that $\zeta_0$ has a representation like in equation \eqref{2.5a}.  For $\zeta_\infty$ consider the   expansion of $e(s;A,A_0)$ for large $v$ given in Lemma \ref{expmean}. Let $J_\infty$ be the largest integer such that $\alpha_{J_\infty}<-\frac{3}{2}$, and write
\beq\label{oo1}
\begin{aligned}
\zeta_\infty(s;A,A_0)=&\int_1^\infty v^{-2s}\left( \sum_{j=0}^{J_\infty}  \sum_{h=0}^{H_j} e_{j,h} v^{2\alpha_j+1}\log^h v \right)dv\\
&+\int_1^\infty v^{-2s} \left(e(s;A,A_0)-\sum_{j=0}^{J_\infty} \sum_{h=0}^{H_j} e_{j,h} v^{2\alpha_j+1}\log^h v\right) dv.
\end{aligned}
\eeq

The last integral in equation (\ref{oo1}) is convergent, while the first one can be computed explicitly. This gives the statement for $\zeta_\infty$, in the sense that $\zeta_\infty$ has a representation like in equation \eqref{2.5a}. Putting together the representations of $\zeta_0$ and $\zeta_\infty$ concludes the proof.

\end{proof}

\begin{corollary}\label{c:2.6} Let $(A,A_0)$ be a pair of non negative self adjoint operators as in Lemma \ref{z1}.  With the notation of that lemma,
\begin{align*}
\Rd_{s=-\frac12}\zeta(s;A,A_0)=&\frac{e_{a,1}}{4},\\
\Ru_{s=-\frac12}\zeta(s;A,A_0)=&\frac{e_{a,0}}{2}-\frac{c_b}{2},\\
\Rz_{s=-\frac12}\zeta(s;A,A_0)=&\frac{1}{2}\sum_{j=0, j\not=b}^{J_0}\frac{c_j}{\beta_j+\frac32}
+\sum_{j=0, j\not=a}^{J_\infty}\sum_{h=0}^{H_j} \frac{(-1)^{h+1}e_{j,h}}{2^{h+1} (\alpha_j+\frac32)^{h+1}}\\
&+\int_0^1 v \bigg(e(-\frac12;A,A_0)-\sum_{j=0}^{J_0} c_j v^{2\beta_j+1}\bigg) dv\\
&+\int_1^\infty v \bigg(e(-\frac12;A,A_0)-\sum_{j=0}^{J_\infty} \sum_{h=0}^{H_j} e_{j,h} v^{2\alpha_j+1}\log^h v\bigg) dv,
\end{align*}
where in the lower limits of the sums $a$ is the index in the sequence $\{\alpha_j\}$, such that $\alpha_a=-\frac{3}{2}$, and $b$ is the index in the sequence $\{\beta_j\}$, such that  $\beta_b=-\frac{3}{2}$, and
\[\begin{aligned}
\Ru_{s=0}\zeta(s;L,L_0)=&-\frac{e_{a,1}}{4}\beta ,\\
\Rz_{s=0}\zeta(s;L,L_0)=&-\frac{1}{2}\left( c_b-e_{a,0}-(1-\log 2) e_{a,1}\right) \beta,\\
\Rz_{s=0}\zeta'(s;L,L_0)=&-\beta\Bigg(\frac{1}{2}\sum_{j=0, j\not=b}^{J_0}\frac{c_j}{\beta_j+3/2}
+\sum_{j=0, j\not=a}^{J_\infty}\sum_{h=0}^{H_j} \frac{(-1)^{h+1}e_{j,h}}{2^{h+1} (\alpha_j+3/2)^{h+1}}\\
&+\int_0^1 v \bigg(e(-1/2;A,A_0)-\sum_{j=0}^{J_0} c_j v^{2\beta_j+1}\bigg) dv\\
&+\int_1^\infty v \bigg(e(-1/2;A,A_0)-\sum_{j=0}^{J_\infty} \sum_{h=0}^{H_j} e_{j,h} v^{2\alpha_j+1}\log^h v\bigg) dv\Bigg)\\
&-\beta (1-\log 2)(e_{a,0}-c_b)\\
&-\beta\bigg(\frac12 +\frac{\pi^2}{24}+\frac{(1-\log 2)^2}{2}  \bigg) e_{a,1} \\
&-2\log\eta(\beta ;A,A_0).
\end{aligned}\]
\end{corollary}
\begin{proof}
This is a simple consequence of Lemma \ref{rzf} and Prop. \ref{p1}.
\end{proof}

\section{Coulomb potential plus delta interaction centered at the origin}
\label{s1}

\subsection{Preliminaries}

Recall that  we denote by $\rho(A)$ the resolvent set of $A$ and
by $R(\lambda;A)$ the resolvent operator $(\lambda I-A)^{-1}$, for
$\lambda\in \rho(A)$. If $R(\la;A) $ operates in $L^2(\RE^3)$, we denote by $k(\lambda; A)=k(\lambda;A)(x,y)$ the integral kernel of $R(\lambda; A)$, $x,y\in\RE^3$.

Let $H_0$ be the self-adjoint realization of the  operator
$-\Delta+\ga/|x|$ in $L^2(\RE^3)$, namely the Laplace operator plus a
Coulomb potential centered at the origin in the three dimensional
Euclidean space, with parameter $\gamma\in\RE$. The kernel of the resolvent of $H_0$ is, see, e.g.,   \cite[eq. (2.1.16)]{AGH-KH05} and  \cite{BGPjmp4605,BS},
\[k(\la;H_0)(x,y)=
-\frac{\Gamma(1+\ga/2\sqrt{-\la})}{4\pi|x-y|} f_{\lambda}(x,y),
\]
where
\begin{align*}
f_{\lambda}(x,y)= & W_{-\frac{\gamma}{2\sqrt{-\lambda}},\frac{1}{2}}(\sqrt{-\la}\ x_+)M'_{-\frac{\gamma}{2\sqrt{-\lambda}},\frac{1}{2}}(\sqrt{-\la}\ x_-) \\
&-W'_{-\frac{\gamma}{2\sqrt{-\lambda}},\frac{1}{2}}(\sqrt{-\la}\ x_+)M_{-\frac{\gamma}{2\sqrt-{\lambda}},\frac{1}{2}}(\sqrt{-\la}\ x_-),
\end{align*}
with $\la\in\rho(H_0)$, $\Re\sqrt{-\la}>0$, $x_\pm=|x|+|y|\pm |x-y|$, and where
$M_{\ka,\mu}$ and $W_{\ka,\mu}$ are  Whittaker functions, see e.g. \cite{GraRyz07}.
In the next proposition we recall some results on the spectrum of $H_0$,
see e.g. \cite{AGH-KH05, Hosjmp867}.
\begin{proposition}
For all $\gamma\in\RE$ the essential spectrum of $H_0$ is purely absolutely continuous, moreover
\[
\sigma_{ess}(H_0)=\sigma_{ac}(H_0)=[0,+\infty)\,.
\]
If $\gamma\geq 0 $ the point spectrum of  $H_0$ is empty. If $\gamma <0$ the point spectrum of $H_0$ is
\[
\sigma_{pp}(H_0)=\Big\{-\frac{\gamma^2}{4(n+1)^2}\Big\}_{n=0}^{\infty}\qquad \gamma<0\,.
\]
\end{proposition}
%



Following \cite{AGH-KH05}, we introduced a perturbation of $H_0$, by adding a  singular one center point interaction, also centered at the origin.  For all $-\infty<\al \leq\infty$ we denote by $H_\al$ the operator formally written as $-\Delta+\gamma/|x|+\al\delta_0$. The concrete operator is defined in Theorem 2.1.2 of \cite{AGH-KH05}, and the  integral kernel of the resolvent of $H_\al$ is
\beq
\label{RHal}
\begin{aligned}
k(\la;H_\al)(x,y)=&k(\la;H_0)(x,y)
-
\frac{4\pi}{4\pi\al-\ga F_\gamma(\ga/2\sqrt{-\la})}\gg(\la;x)\gg(\la;y)\,,\\
&\la\in\rho(H_\al)\cap\rho(H_0)\,,\;\Re\sqrt{-\la}>0\,,\; x,y\in\RE^3
\end{aligned}
\eeq
with
\[
\gg(\la;x):=\frac{\Gamma(1+\ga/2\sqrt{-\la})}{4\pi|x|}
\WW_{-\ga/2\sqrt{-\la},1/2}(2\sqrt{-\la} \ |x|)\qquad x\neq0,
\]
and
\[
\begin{aligned}
F_\gamma(z): =&
\begin{cases}
\psi(1+z)-\log(z)-\frac{1}{2z}-\psi(1)-\psi(2),&\gamma>0,\\
\psi(1+z)-\log(-z)-\frac{1}{2z}-\psi(1)-\psi(2),&\gamma<0.
\end{cases}
\end{aligned}
\]
Here $z\in\CO$ and  $\psi$ is the digamma function, i.e., $\psi(z) = \frac{d}{dz} \log \Gamma(z)$, see \cite[8.36]{GraRyz07}. We note that the function $F_\ga(z)$ is indeed a function of $z$ and of
$\sgn(\ga)$ only.

In the following proposition we recall some results on
the spectrum of the operator $H_\alpha$, see, e.g. \cite[Th. 2.1.3]{AGH-KH05}.
\begin{proposition} Let $-\infty<\al \leq\infty$.
For all $\gamma\in\RE$ the essential spectrum of $H_\alpha$ is purely absolutely continuous, moreover
\[
\sigma_{ess}(H_\alpha)=\sigma_{ac}(H_\alpha)=[0,+\infty)\,.
\]
The eigenvalues of $H_\alpha$ associated with the $s$-wave $(l=0)$ are given by the solutions of the equation
\beq
\label{eigeneqganeg}
4\pi\al-\ga F_\ga(\gamma/2\sqrt{-E})=0,\qquad \;E<0\,,
\eeq
where we set $\ga F_\ga(\gamma/2\sqrt{-E})\big|_{\gamma=0}=
\lim_{\ga\to0}\ga F_\ga(\gamma/2\sqrt{-E})=-\sqrt{-E}$.

 If $\gamma\geq0$ and $\alpha\geq
 -\gamma[\psi(1)+\psi(2)]/4\pi$, the  equation \eqref{eigeneqganeg} has
 no solutions, moreover the point spectrum of $H_\alpha$ is empty.

If $\gamma\geq0$  and  $\alpha<-\gamma[\psi(1)+\psi(2)]/4\pi$, the equation
 \eqref{eigeneqganeg} has precisely one solution, and the operator   $H_\alpha$ has
 precisely one negative eigenvalue.

If $\gamma<0$ the equation \eqref{eigeneqganeg} has infinitely many
 solutions. Correspondingly there are infinitely many simple eigenvalues associated
 with the $s$-wave $(l=0)$, moreover for $l\geq 1$ the eigenvalues of $H_\alpha$ are given by the usual Coulomb levels $E_m=-\gamma^2/4m^2$, $m\in \NO$, $m\geq 2$.
\end{proposition}

Because of the results of the previous proposition, we proceed our analysis only in the case of repulsive Coulomb potential, namely for $\gamma\geq 0$.

\subsection{Trace of the relative resolvent}
We first note that for any $\la\in
\rho(H_0)\cap\rho(H_\al)$  the difference $\tr(R(\la;H_\al)-R(\la;H_0))$ is a
rank one operator (see, e.g. \cite{AGH-KH05}), then  the trace  of the relative resolvent of the pair $(H_\al,H_0)$ is well defined by
\[
r(\la;H_\al,H_0)=\tr(R(\la;H_\al)-R(\la;H_0))\,.
\]

By equation  (\ref{RHal}) and by the definition of $\gg(\la,x)$ it
follows that
\[
\begin{aligned}
r(\la;H_\al,H_0)=
-\frac{\Gamma(1+\ga/2\sqrt{-\la})^2}{4\pi\al-\ga F_\ga(\ga/2\sqrt{-\la})}
\int_0^\infty
\WW_{-\ga/2\sqrt{-\la},1/2}^2(2\sqrt{-\la} \ |x|)\,
 d|x|
\end{aligned}
\]with $\Re\sqrt{-\lambda}>0$.

From  \cite[$7.625.4$ and $9.302.1$]{GraRyz07}  and by using the
identities  $\Gamma(1+z)=z\Gamma(z)$ and
$\Gamma(1-z)\Gamma(z)=\frac{\pi}{\sin(\pi z)}$, $z\in\CO\backslash\ZA$,  we get for the integral
in the latter equation the expression
\[
\begin{aligned}
&\int_0^\infty\WW_{-\ga/2\sqrt{-\la},1/2}^2(2\sqrt{-\la} \ |x|)\,d|x| \\
=&
\frac{1}{2\sqrt{-\la}} \frac{1}{\Gamma(1+\ga/2\sqrt{-\la})
\Gamma(\gamma/2\sqrt{-\la})}\\
&\times
\frac{1}{2\pi i}\int_L
\frac{(1-s)s}{(s+\ga/2\sqrt{-\la})(s-1+\ga/2\sqrt{-\la})}\frac{\pi^2}{\sin^2(\pi s)} ds,
\end{aligned}
\]
where $L$ is a path  in $\CO$ from $-\infty$ to $+\infty$ such that the set
$\{1,2,3,...\}$ is on the right of $L$ and the set
$\{0,-1,-2,...,1-\ga/2\sqrt{-\la},
-\ga/2\sqrt{-\la}, -1-\ga/2\sqrt{-\la}, ...  \}$  is on the left of
$L$. We notice that for $\ga>0$ one can choose $L=\{s=x_0+iy\;
\textrm{with}\;1-\ga\Re\sqrt{-\la}/2|\la|<x_0<1\,,\;-\infty<y<\infty\}$.

This gives
\beq
\label{reltr}
r(\la;H_\al,H_0)=-
\frac{1}{4\pi\al-\ga F_\ga(\ga/2\sqrt{-\la})}
\frac{1}{2\sqrt{-\la}} I(\ga/2\sqrt{-\la}),
\eeq
where
\beq
\label{Izeta}
I( z )=\frac{ z }{2\pi i}
\int_L
\frac{(1-s)s}{(s+ z )(s-1+ z )}\frac{\pi^2}{\sin^2(\pi s)} ds,
\eeq
and we used again the identity $\Gamma(1+z)=z\Gamma(z)$, $z\in\CO$. In order to analyze the function $I(z)$ appearing in the formula for the relative trace of the resolvent  we need the formulas in the following lemma.
\begin{lemma}
\label{prop-ints}
Let $L$ be the path $L=\{z=x_0+iy|\,1-a<x_0<1\,,\;-\infty<y<\infty\}$, with $\Re(a)>0$, then
\[
\frac{1}{2\pi i}\int_{L}\frac{\pi^2}{\sin^2\pi z}dz=1\,,
\]
and 
\[
\frac{1}{2\pi i}\int_{L}\frac{1}{z+a}\frac{\pi^2}{\sin^2\pi z}dz=\psi'(a)-\frac{1}{a^2}.
\]
\end{lemma}
\begin{proof} For the first, we  just integrate
\begin{align*}
&\frac{1}{2\pi i}\int_{L}\frac{\pi^2}{\sin^2\pi z}dz=-\frac{1}{2i} \left[ \cot(\pi z)\right]_{x_0-i y}^{x_0+i y} \\
=&
\lim_{y\to \infty}-\frac{1}{2}\left( \frac{\e^{ix_0- y}+\e^{-ix_0+ y}}{\e^{ix_0- y}-\e^{-ix_0+ y}}-\frac{\e^{ix_0+ y}+\e^{-ix_0-iy}}{\e^{ix_0+ y}-\e^{-ix_0- y}}\right)=1.
\end{align*}

For the second one, we first integrate twice by parts. This gives
\[
\int_{L}\frac{1}{a+z}\frac{\pi^2}{\sin^2\pi z}dz=-2\int_L \frac{\log\sin \pi z}{(a+z)^3} dz.
\]

Next,we  use the product representation for the sine function:
\[
\int_L \frac{\log\sin \pi z}{(a+z)^3}dz=\int_L \frac{\log\pi
 z}{(a+z)^3}dz +
\int_L \frac{1}{(a+z)^3} \sum_{k=1}^\infty \log\left(1-\frac{z^2}{k^2}\right) dz.
\]

The first term gives no contribution, for
\begin{align*}
\int_L \frac{\log\pi z}{(a+z)^3}dz&=-\frac{1}{2}\left[ \frac{\log\pi z}{(a+x_0+iy)^2}\right]_{y=-\infty}^{y=+\infty} +\frac{1}{2}\int_L \frac{1}{z(a+z)^2} dz\\
&=0+\frac{1}{2}\left[\frac{1}{a(a+z)}-\frac{1}{a^2} \log\left(1+\frac{a}{z}\right)\right]_{y=-\infty}^{y=+\infty}=0.
\end{align*}

In the second term, due to uniform convergence, we can twist the sum with the integration. We have
\[
\int_L \frac{1}{(a+z)^3} \log\left(1-\frac{z^2}{k^2}\right) dz=0-\int_L \frac{z}{(a+z)^2}\frac{1}{k^2-z^2} dz.
\]
Assuming $\Re(a)>0$, we can deform the path $L$ to a contour of Hankel type: starting at infinity on the upper side of the real axis, turning around the point $z=k$ and going back to infinity below the real axis. Since the integrand  vanishes as $z^{-3}$ for large $\Re(z)$, we can further deform the path of integration to a circle around the point $z=k$. This gives
\[
\int_L \frac{z}{(a+z)^2}\frac{1}{k^2-z^2} dz=-\frac{\pi i}{(a+k)^2},
\]
and hence the second formula of the lemma follows recalling the definition of the digamma function $\psi(z)$, see \cite[8.36]{GraRyz07}.
\end{proof}
Now we can use the result of the latter lemma to give an explicit
expression for the function $I(z)$.
\begin{lemma}
\label{prop-I}
Let  $I( z )$ be the function defined in equation \eqref{Izeta}, then
\[
I( z )=1-2 z +2\psi'(1+ z ) z ^2\,.
\]
\end{lemma}
\begin{proof} We observe that
\[
\frac{(1-s)s}{(s+ z )(s-1+ z )}=-1+\frac{ z (1+ z )}{s+ z }+\frac{ z (1- z )}{s-1+ z }\,,
\]
from which it follows that
\[
\begin{aligned}
\frac{ z }{2\pi i}
\int_{L}\frac{(1-s)s}{(s+ z )(s-1+ z )}\frac{\pi^2}{\sin^2\pi s} ds
=&-
\frac{ z }{2\pi i}
\int_{L}\frac{\pi^2}{\sin^2\pi s} ds\\
&+\frac{ z ^2(1+ z )}{2\pi i}
\int_{L}
\frac{1}{s+ z }
\frac{\pi^2}{\sin^2\pi s} ds\\
&+\frac{ z ^2(1- z )}{2\pi i}
\int_{L}
\frac{1}{s-1+ z }
\frac{\pi^2}{\sin^2\pi s} ds.
\end{aligned}
\]

Using Lemma \ref{prop-ints}, and recalling that
 $\psi(z+1)=\psi(z)+\frac{1}{z}$, after some calculation we have the
 stated formula.
\end{proof}


\begin{proposition}
\label{p.rel-tr}
For any $\la\in\rho(H_\al)\cap\rho(H_0)$, the trace of the relative resolvent of  the pair of operators $(H_\alpha,H_0 )$ is given
 by
\beq
\label{r-cou}
\begin{aligned}
 r(\lambda;H_\al,H_0)=
&-\frac{z I(z)}{\ga(4\pi\al-\ga F_\ga(z))}\bigg|_{z=\frac{\ga}{2\sqrt{-\la}}}\\
 =&-\frac{z(2\psi'(1+z)z^2-2z+1)}{\ga(4\pi\al-\ga(\psi(1+z)-\log z -\frac{1}{2z}-\psi(1)-\psi(2)))}
\bigg|_{z=\frac{\ga}{2\sqrt{-\la}}}\,,
\end{aligned}
\eeq
with $\Re\sqrt{-\la}>0$. Moreover the following asymptotic expansion
 holds true for small $\lambda$,
\beq
\label{smalllambda}
r(\lambda;H_\al,H_0)=\sum_{k=0}^\infty b_k (-\lambda)^{k},
\eeq
with $b_k\in\RE$. The   first coefficients are given by
\[
b_0=-\frac{1}{3\ga (\ga-2 C\ga +4\pi \al)}, \hspace{30pt}
b_1=\frac{(17-24 C)\ga+48\pi\al}{45\ga^3(2C\ga-\ga -4\pi\al)^2}.
\]
For large $\lambda$, we have the following asymptotic expansion
\beq
\label{largelambda}
r(\lambda;H_\al,H_0)=\sum_{j=2,k=0}^\infty a_{j,k} (-\lambda)^{-\frac{j}{2}}\log^k (-\la),
\eeq
with $a_{j,k}\in\RE$. The  first coefficients are given by
\[
a_{2,0}=-\frac{1}{2}, \hspace{10pt} a_{2,k>0}=0\]
\[ a_{3,0}=\frac{4\pi\al+(2 -C)\ga+\ga(\log \ga-\log 2)}{2},\hspace{10pt} a_{3,1}=-\frac{\ga}{4}, \hspace{10pt} a_{3,k>1}=0.
\]
\end{proposition}
\begin{proof}
Formula \eqref{r-cou} follows directly from the equation \eqref{reltr}
 and from Lemma \ref{prop-I}. The  asymptotic expansions of the relative trace follow easily from classical expansion of the poly Gamma function. Recalling the expansions of the digamma function (see for example \cite[8.344]{GraRyz07}) for large $|z|$, $z\in \CO$
\[
\psi(1+z)=\log z +\frac{1}{2}\frac{1}{z}-\sum_{k=1}^\infty \frac{B_{2k}}{2k}\frac{1}{z^{2k}},
\]
$B_{2k}$ being Bernoulli numbers,
and for small $z\in\CO$ (see for example \cite[8.342]{GraRyz07}):
\[
\psi(1+z)=-C+\sum_{k=2}^\infty (-1)^k \zeta(k) z^{k-1},
\]
where $C$ is the Euler constant, and where $\zeta$ denotes the Riemann's zeta function. Since
 $z=\frac{\gamma}{2\sqrt{-\lambda}}$ one obtains the expansions \eqref{smalllambda}
 and \eqref{largelambda}.
\end{proof}

\section{The relative partition function of the Coulomb plus delta interaction}
\label{s3}

In this section we study the relative zeta function and the relative partition function of the model described in Section \ref{s1}.

In order to obtain the relative zeta function for the pair of operators $(H_\alpha,H_0)$ described in Section \ref{s1}. It is clear, by the result in Prop. \ref{p.rel-tr}, that the conditions (B.1), (B.2) and (B.3) of Lemma \ref{l2.1}, necessary to define the relative zeta function are satisfied. Also, by the same proposition, the minimum value for the index $j$ is $j=2$, corresponding to $\alpha_2=-1$, then the first terms in the expansion of the relative spectral measure, according to  Lemma \ref{expmean}, are: first, the term corresponding to $\alpha_2=-1$, that gives a only a term in $\frac{1}{v}$, since $K_2=0$; second,  the terms corresponding to $\alpha_3=-\frac{3}{2}$, that gives a term in $\frac{1}{v^2}$ and a term in $\frac{1}{v^2}\log v^2$, since $K_{3}=1$. Applying the formula in Lemma \ref{expmean}, the coefficients are:
\begin{align*}
e_{2,0,0}&=0;\\
e_{3,0,0}&=\frac{4\pi\alpha+(2-C)\gamma+\gamma\log\frac{\gamma}{2}}{\pi},&e_{3,1,0}&=0,&e_{3,1,1}&=-\frac{\gamma}{2\pi}.
\end{align*}

Whence, we have the following expansion of the relative spectral measure:
\begin{equation}\label{exprsm1}
e(v;H_\alpha, H_0)=O(v^k), ~~ k>0
\end{equation}
for $v\to 0^+$,
\begin{equation}
\label{exprsm2}
e(v;H_\alpha, H_0)=-\frac{\gamma}{\pi}\frac{1}{v^2}\log v+\frac{4\pi\alpha+(2-C)\gamma+\gamma\log\frac{\gamma}{2}}{\pi} \frac{1}{v^2}+O(v^{-3}\log v)
\end{equation}
for $v\to +\infty$,  and $e_{3,0}=\frac{4\pi\alpha+(2-C)\gamma+\gamma\log\frac{\gamma}{2}}{\pi}$, $e_{3,1}=-\frac{\gamma}{\pi}$. All the coefficients $e_{j,h}$ with smaller indices vanish.

We are now in the position of analyzing the relative zeta function $\zeta(s;H_\alpha,H_0)$. In fact, what we are interested in is the expansion near $s=-\frac{1}{2}$.

\begin{proposition}\label{p:4.1} The relative zeta function $\zeta(s;H_\alpha,H_0)$ has an analytic expansion to a meromorphic function analytic in the strip $0\leq \Re(s)\leq 1$, up to a double pole at  $s=-\frac{1}{2}$. Near $s=-\frac{1}{2}$, the following expansion holds:
\begin{align*}
\zeta(s;H_\alpha,H_0)=&\frac{e_{3,1}/4}{\left(s+\frac{1}{2}\right)^2}+\frac{e_{3,0}/2}{s+\frac{1}{2}}+\int_0^1 v e(v;H_\al,H_0) dv \\
&+\int_1^\infty v\left(e(v;H_\al,H_0)-\frac{e_{2,1}}{v^2}\log v-\frac{e_{2,0}}{v^2}\right) dv +O\left(s+\frac12\right),
\end{align*}
where
\[
e_{3,1}=-\frac{\gamma}{\pi},\hspace{30pt} e_{3,0}=\frac{8\pi\alpha-2 C\ga+4\ga+\ga\log\frac{\ga^2}{4}}{2\pi}.
\]
\end{proposition}
\begin{proof} By the expansions in equations \eqref{exprsm1} and \eqref{exprsm2} for the relative spectral measure, we see that the indices $J_0$ and $J_\infty$ defined in Lemma \ref{rzf} are respectively: $J_0=0$ and $J_\infty=4$. Hence, by the same lemma, there are no poles arising from the expansion of the spectral measure for small $v$, and since the minimum value for the index $j$ of $\alpha_j$ is $j=2$, there are three terms arising from the expansion for large $v$. The first term is with $j=2$, and vanishes since $e_{2,0}=0$. The other two terms are with $j=3$, and $k=0$ and $k=H_3=1$. Applying the formula in Lemma \ref{rzf} we compute these terms.
\end{proof}

\begin{corollary} The relative zeta function of the pair of operators $(L=-\p^2_u+H_\al, L_0=-\p_u^2+H_0)$ on $S^1_\frac{\beta}{2\pi}\times \R^3$ has a simple pole at $s=0$ with residua:
\begin{align*}
&\Ru_{s=0}\zeta(s;L,L_0)=\frac{\gamma}{4\pi}\beta, \\
&\Rz_{s=0}\zeta(s;L,L_0)=
-\frac{8\pi\alpha-2 C\ga+4\ga+\ga\log\frac{\ga^2}{4}}{4\pi}\beta+\frac{(1-\log 2)\gamma}{2\pi}\beta, \\
&\Rz_{s=0}\zeta'(s;L,L_0)\\
=&-\left(\int_0^1 v e(v;H_\al,H_0) dv+\int_1^\infty v\left(e(v;H_\al,H_0)-\frac{e_{3,1}}{v^2}\log v-\frac{e_{3,0}}{v^2}\right) dv\right)\beta\\
&- \frac{(1-\log 2)(8\pi\alpha-2 C\ga+4\ga+\ga\log\frac{\ga^2}{4})}{2\pi}\beta
+\left(2+\frac{\pi^2}{6}+2(1-\log 2)^2\right)\frac{\gamma}{4\pi}\beta\\
&-2\int_0^\infty \log\left(1-\e^{-\beta v}\right) e(v;H_\al,H_0) dv.
\end{align*}

\end{corollary}
\begin{proof}
This is a simple consequence of Prop. \ref{p:4.1} and Cor. \ref{c:2.6}.
\end{proof}

Using the formula in equation (\ref{partpart}), we obtain the following result for the relative partition function, where $\ell$ is some renormalization constant,
\begin{align*}
& \log Z_{\rm R} \\
=&-\frac{1}{2}\left(\int_0^1 v e(v;H_\al,H_0) dv+\int_1^\infty v\left(e(v;H_\al,H_0)-\frac{e_{3,1}}{v^2}\log v-\frac{e_{3,0}}{v^2}\right) dv\right)\beta\\
&- \frac{(1-\log 2)(8\pi\alpha-2 C\ga+4\ga+\ga\log\frac{\ga^2}{4})}{4\pi}\beta
+\left(2+\frac{\pi^2}{6}+2(1-\log 2)^2\right)\frac{\gamma}{8\pi}\beta\\
&-\int_0^\infty \log\left(1-\e^{-\beta v}\right) e(v;H_\al,H_0) dv \\
&+\left(4\pi\alpha- C\ga+2\ga+\ga\log\frac{\ga}{2}-1+\gamma\log 2\right)\frac{\beta}{2\pi}
\log \ell.
\end{align*}

\vspace{30pt}
{\bf Acknowledgments.}  The authors are very pleased to dedicate this work to Pavel Exner, on the occasion of his 70th birthday. He has always been for us a source of inspiration, and we are very grateful to him for his support. We thank the referee for the very useful suggestions and for pointing out additional references. The authors acknowledge the hospitality and financial support of the CIRM, University of Trento, Italy.  C.C. also gratefully acknowledges the hospitality of the  Mathematical Institute, Tohoku University, Sendai, Japan, of  the Hausdorff Center for Mathematics, University of Bonn, Germany, and the support of the FIR 2013 project, code RBFR13WAET.


\begin{thebibliography}{1}
\bibitem{0} S. Albeverio, \emph{Wiener and Feynman--path integrals and their applications}, Proceedings of the Norbert Wiener centenary congress, AMS Press, 1997.

\bibitem{ABD}
S. Albeverio, Z. Brze\'{z}niak, and L. Dabrowski,
\emph{Fundamental solution of the heat and Schr{\"o}dinger equations with point interaction},
J. Funct. Anal. \textbf{130} (1995), 220--254.

\bibitem{ACSZ}
S. Albeverio, G. Cognola, M. Spreafico, and S. Zerbini, \emph{Singular perturbations with boundary conditions and the
              {C}asimir effect in the half space}, J. Math. Phys. \textbf{51} (2010), 063502, 38pp.

\bibitem{AGH-KH05}
S.~Albeverio, F.~Gesztesy, R.~{H{\o}egh-Krohn}, and H.~Holden, \emph{Solvable
  models in quantum mechanics: {S}econd edition}, AMS Chelsea Publ., 2005, with
  an Appendix by P. Exner.

  \bibitem{AK}
S. Albeverio, P. Kurasov,
\emph{Singular perturbations of differential operators},
Cambridge University Press (2000).

\bibitem{AsIM}
M. Asorey, A. Ibort, and G. Marmo,
\emph{Global theory of quantum boundary conditions and topology change},
Int. J. Mod. Phys. A \textbf{20} (2005), 1001--1026.

\bibitem{AtPaSi}
M.F. Atiyah, V.K. Patodi, and I.M. Singer,
\emph{Spectral asymmetry and Riemannian geometry, III},
Math. Proc. Cambridge Phil. Soc. \textbf{79} (1976), 71--99.

\bibitem{Beth-Uhl}
E. Beth, G.E. Uhlenbeck, \emph{The quantum theory of the non-ideal gas. II. Behaviour at low temperatures},
Physica \textbf{4} (1937), 915--924.

\bibitem{birman-krein} M.\v{S}. Birman, M. G. Krein, \emph{On the theory of wave operators and scattering operators}, Soviet Math. Dokl \textbf{3} (1962), 740--744.

\bibitem{Boi}
L. Boi,
\emph{The Quantum Vacuum},
The J. Hopkins Univ. Press, Baltimore (2011).

\bibitem{Bordagetal}
M. Bordag, ed.,
\emph{The {C}asimir effect 50 years later},
in ``Proceedings the IV Workshop on quantum field theory under the influence of external conditions'', Leipzig, 1998, World Scient., Singapore (1999).

\bibitem{BKMM}
M. Bordag, G.L. Klimchitskaya, U. Mohideen, and V.M. Mostepanenko,
\emph{Advances in the {C}asimir effect},
Oxford UP (2009).

\bibitem{BNP02} M. Bordag, V.V. Nesterenko, and I.G. Pirozhenko,  \emph{High temperature asymptotics of thermodynamic functions of an electromagnetic field subjected to boundary conditions on a sphere and cylinder}, Physical Review D \textbf{65}  (2002), no. 4, 045011.

\bibitem{BPN}
M. Bordag, I.G. Pirozhenko, and V.V. Nesterenko,
\emph{Spectral analysis of a flat plasma sheet model},
J. Phys. A \textbf{38} (2005), 11027.

\bibitem{BV}
M. Bordag, D.V. Vassilevich,
\emph{Heat kernel expansion for semitransparent boundaries},
J. Phys. A \textbf{32} (1999), 8247--8259.

\bibitem{BGPjmp4605}
J.~Br{\"u}ning, V.~Geyler, and K.~Pankrashkin, \emph{On-diagonal singularities
  of the {G}reen functions for Schr{\"o}dinger operators}, J. Math. Phys.
  \textbf{46} (2005), 113508, 16pp.

\bibitem{BS}
J. Br{\"u}ning, R. Seeley,
\emph{The resolvent expansion for second-order regular singular operators},
J. Funct. Analysis \textbf{73} (1987), 369--429.

  \bibitem{Bressietal}
G. Bressi, G. Carugno, R. Onofrio, and G. Ruoso,
\emph{Measurement of the {C}asimir force between parallel metallic plates},
Phys. Rev. Lett. \textbf{88} (2002).


\bibitem{BEORZ}
A.A. Bytsenko, E. Elizalde, S. Odintsov, A. Romeo, and S. Zerbini,
\emph{Zeta-function regularization with applications},
World Scientific (1994).

\bibitem{casimir}
H.B.G. Casimir, \emph{On the attraction between two perfectly conducting plates}, Proceedings of the Royal Netherlands Academy of Arts and Sciences \textbf{51} (1948), 793--795.

\bibitem{Cas}
H.G.B. Casimir,
\emph{On the attraction of two perfectly conducting plates},
Proc. Kongl. Nedel. Akad. Wetensch \textbf{51} (1948), 793--795.

\bibitem{CQ}
J. Choi, J.R. Quine,
\emph{Zeta regularized products and functional determinants on spheres},
Rocky Mt. J. Math. \textbf{26} (1996), 719--729.

\bibitem{Cognolaetal}
G. Cognola, L. Vanzo, and S. Zerbini
\emph{Vacuum energy in arbitrarily shaped cavities}
J. Math. Phys. \textbf{33} (1992), 222--228.

\bibitem{ColindeVerdiere}
Y. Colin de Verdi\`{e}re,
\emph{Pseudo-Laplaciens II},
Ann. Inst. Fourier, Grenoble \textbf{33} (1983), 87--113.

\bibitem{CFNT}
M.V. Cougo Pinto, C. Farina, M.R. Negr\~{a}o, and A. Tort,
\emph{{C}asimir effect at finite temperature of charged scalar field in an external magnetic field}, arXiv:hep-th/9810033
(1998), 6pp.

\bibitem{DoKe}
J.S. Dowker, G. Kennedy,
\emph{Finite temperature and boundary effects in static space-times},
J. Phys. A \textbf{11} (1978), 895.

\bibitem{Do}
J.P. Dowling,
\emph{The mathematics of the {C}asimir effect},
Math. Mag. \textbf{62} (1989), 324--331.

\bibitem{Duk}G.V. Dunne, K. Kirsten, \emph{Simplified vacuum energy expressions for radial backgrounds
              and domain walls}, J. Phys. A: Math. and Theo. \textbf{42} (2009), no. 7, 075402.


\bibitem{Elizalde1}
E. Elizalde,
\emph{Ten physical applications of spectral zeta functions},
Springer, Berlin (1995).

\bibitem{Elizalde}
E. Elizalde, S.D. Odintsov, A. Romeo, A.B. Bytsenko, and S. Zerbini,
\emph{Zeta Regularization Techniques and Applications},
World Scient., Singapore (1999).

\bibitem{EVZ}
E. Elizalde, L. Vanzo, and S. Zerbini,
\emph{Zeta-function regularization, the multiplicative anomaly and the Wodzicki Residue},
Comm. Math. Phys. \textbf{194} (1998), 613--630.

\bibitem{Erd54}A. Erd\'elyi, et al., eds. Tables of Integral Transforms: Vol.: 2. McGraw-Hill, 1954.

\bibitem{FeP1}
D. Fermi, L. Pizzocchero,
\emph{Local zeta regularization and the {C}asimir effect},
Progr. Theor. Phys. \textbf{126} (2011), 419--434.

\bibitem{FeP2}
D. Fermi, L. Pizzocchero,
\emph{Local zeta regularization and the {C}asimir effect I-IV}, arXiv:1505.00711, 1505.01044, 1505.01651, 1505.03276 [math-ph] (2015).

\bibitem{Fie}
M. Fierz,
\emph{On the attraction of conducting planes in vacuum},
Helv. Phys. Acta \textbf{33} (1960), 855.

\bibitem{FHKN}
R. Figari,  R.  H\"oegh-Krohn, and C.R. Nappi, \emph{Interacting relativistic boson fields in the de Sitter universe with two space-time dimensions} Comm. Math.  Phys. \textbf{44} (1975), 265--278.

\bibitem{Ful}
S.A. Fulling,
\emph{Aspects of quantum field theory in curved space-time},
Cambridge Univ. Press (1989).

\bibitem{Gibbons}
G. Gibbons,
\emph{Thermal zeta functions},
Phys. Letters A \textbf{60} (1977), 385--386.

\bibitem{Gilkey}
P.B. Gilkey,
\emph{Invariance theory, the heat equation, and the Atiyah-Singer index theorem},
Studies in Advanced Mathematics (CRC Press, 1995).

\bibitem{GKV}
P.B. Gilkey, K. Kirsten, and D.V. Vassilevich,
\emph{Heat trace asymptotics defined by transfer boundary conditions},
Lett. Math. Phys. \textbf{63} (2003), 29--37.

\bibitem{GraRyz07}
I.~S. Gradshteyn and I.~M. Ryzhik, \emph{Table of integrals, series, and
  products. $7^{th}$ ed.}, Academic Press, 2007.

  \bibitem{GJKQSW}
N. Graham, R.L. Jaffe, V. Khemani, M. Quandt, M. Scandurra, and H. Weigel,
\emph{Calculating vacuum energies in renormalizable quantum field theories: A new approach to the {C}asimir problem},
Nucl. Phys. B \textbf{645} (2002), 49--84.

\bibitem{Hawking}
S.W. Hawking,
\emph{Zeta function regularization of path integral in curved spacetime},
Comm. Math. Phys. \textbf{55} (1977), 133--148.

\bibitem{Hosjmp867}
L.~Hostler, \emph{Runge-{L}enz vector and the {C}oulomb {G}reen's functions},
  J. Math. Phys. \textbf{8} (1967), no.~3, 642--646.

  \bibitem{Khusnutdinov}
N.R. Khusnutdinov,
\emph{Zeta-function approach to {C}asimir energy with singular potentials},
Phys. Rev. D \textbf{73} (2006), 025003.

\bibitem{krein1}
M.G.~Krein, \emph{On the trace formula in perturbation theory}, Mat. Sb. (N.S.) \textbf{33} (1953), no. 75, 597--626.

\bibitem{krein2}
M.G. Krein, \emph{Perturbation determinants and a formula for the traces of unitary and self-adjoint operators}, Sov. Math. Dokl \textbf{3} (1962), 707--710

\bibitem{Kurokawa}
N. Kurokawa,
\emph{Multiple sine functions and Selberg zeta functions},
Proceedings of the Japan Academy, Series A, Mathematical Sciences \textbf{67} (1991), 61--64.

\bibitem{Lamoureaux}
S.K. Lamoureaux,
\emph{Demonstration of the {C}asimir force in the 0.6 to 6$\mu m$ range},
Phys. Rev. Lett. \textbf{78} (1997), 5--8.

\bibitem{Lesch}
M. Lesch,
\emph{Determinants of regular singular Sturm-Liouville operators},
Math. Nachr. \textbf{194} (1998), 139--170.

\bibitem{Meh}J. Mehra, \emph{Temperature correction to the Casimir effect}, Physica \textbf{37} (1967), no. 1, 145--152.

%

\bibitem{Milton}
K.A. Milton,
\emph{The {C}asimir effect: The Physical Manifestation of Zero-Point Energy},
World Scientific, Singapore (2001).

\bibitem{Mo}
V. Moretti, \emph{Local $\zeta$-function techniques vs. point-splitting procedure: a few rigorous results},
Comm. Math. Phys. \textbf{201} (1999), 327--363.

\bibitem{MT}
V.M. Mostepanenko, N.N. Trunov,
\emph{The {C}asimir Effect and its Applications},
Oxford, Clarendon Press (1997).

\bibitem{Muller}
W. M{\"u}ller,
\emph{Relative zeta functions, relative determinants and scattering theory},
Comm. Math. Phys. \textbf{192} (1998), 309--347.

\bibitem{MuC}
J.M. Mu\~{n}oz Casta\~{n}eda, J. Mateos Guilarte,
\emph{$\delta-\delta'$ generalized Robin boundary conditions and quantum vacuum fluctuations}, Phys. Rev. D \textbf{91} (2015), 025028, 21pp.

\bibitem{NLS04} V.V. Nesterenko, G. Lambiase, and G. Scarpetta, \emph{Calculation of the Casimir energy at zero and finite temperature: Some recent results}, Riv. Nuovo Cimento \textbf{27} (2004), no.6, 1--74.

\bibitem{NP12} V.V. Nesterenko, I.G. Pirozhenko, \emph{Lifshitz formula by a spectral summation method}, Phys. Rev. A \textbf{86} (2012), 052503.


\bibitem{orSpr}
G. Ortenzi, M. Spreafico,
\emph{Zeta function regularization for a scalar field in a bounded domain},
J. Phys. A \textbf{37} (2004), 11499--11517.

\bibitem{Park}
D.K. Park,
\emph{Green's function approach to two and three-dimensional delta function potentials and application to the spin 1/2 Aharonov-Bohm problem},
J. Math. Phys. \textbf{36} (1995), 5453--5464.

\bibitem{PMG}
G. Plunien, B. M{\"u}ller, and W. Greiner,
\emph{The {C}asimir effect},
Phys. Repts. \textbf{134} (1986), 87--193.

\bibitem{Pow}
E.A. Power,
\emph{Introductory Quantum Electrodynamics},
Longman, London (1964).

\bibitem{RS}
D.B. Ray, I.M. Singer,
\emph{R-torsion and the Laplacian on Riemannian manifolds},
Adv. Math. \textbf{7} (1974), 145--210.


\bibitem{Scandurra}
M. Scandurra, \emph{The ground state energy of a massive scalar field in the background of a semi-transparent spherical shell},
J. Phys. A \textbf{32} (1999), 5679.

\bibitem{ST}
S. Scarlatti, A. Teta,
\emph{Derivation of the time dependent propagator for the three-dimensional Schr{\"o}dinger equation with one point interaction},
J. Phys. A \textbf{23} (1990), L1033--L1035.

\bibitem{Solodukhin}
S.N. Solodukhin,
\emph{Exact solution for a quantum field with delta like interaction},
Nucl. Phys. B \textbf{541} (1999), 461--482.

\bibitem{spaarnay}
M.J. Spaarnay
\emph{Measurements of attractive forces between flat plates},
Physica \textbf{24} (1958), 751.

\bibitem{Spreafico1}
M. Spreafico,
\emph{Zeta function and regularized determinant on projective spaces},
Rocky Mt. J. Math. \textbf{33} (2003), 1499--1512.

\bibitem{Spreafico2}
M. Spreafico,
\emph{A generalization of the Euler Gamma function},
Funct. Analysis Applic. \textbf{39} (2005), 87--91.

\bibitem{S06}M. Spreafico, {\it Zeta invariants for sequences of spectral type, special functions and the Lerch formula}. Proceedings of the Royal Society of Edinburgh: Section A Mathematics 136, no. 4, (2006) 863--887.

\bibitem{66a} M. Spreafico, \emph{Zeta determinant and operator determinants}, Osaka Jour. Math.  \textbf{48} (2011),  41--50.

\bibitem{SZ} M. Spreafico,  S. Zerbini, {\it Finite temperature quantum field theory of non compact domain and application to delta interaction},  Rep. Math. Phys. \textbf{63} (2009), 163--177.

\bibitem{Sym}
K. Symanzik, \emph{Schr\"odinger representation and Casimir effect in renormalizable quantum field theory}, Nuclear Phys. B \textbf{190} (1981), no. 1, 1--44.

\end{thebibliography}
\end{document}